\documentclass[letterpaper,11pt]{article}
\usepackage{amsmath, amsthm, amssymb}
\usepackage[usenames, dvipsnames]{color}
\usepackage[normalem]{ulem} 
\usepackage{fullpage}
\usepackage[numbers]{natbib}
\usepackage[noend]{algorithmic}
\usepackage{tikz, pgfplots}
\usepackage{enumerate}
\usepackage{hyperref}
\usepackage{xspace,color}
\usepackage{graphicx}
\usepackage{caption}
\usepackage{subcaption}
\usepackage{enumitem,linegoal}
\usepackage[linesnumbered,ruled,vlined]{algorithm2e}
\usepackage{comment}	
\usepackage{ctable}

\usepackage{nicefrac}
\newtheorem{observation}{Observation}[section]

\usepackage{mathtools}
\usepackage[flushleft]{threeparttable}
\usepackage{verbatim}
\usepackage{setspace}
\usepackage{bbm}
\usepackage{thm-restate}
\usepackage{footnote}
\usepackage{cleveref}
\usepackage{enumitem}
\crefname{LP}{LP}{LPs}

\makesavenoteenv{tabular}
\makesavenoteenv{table}

\usepackage[margin=1in]{geometry}

\definecolor{crimsonglory}{rgb}{0,0,0}

 \newtheorem{theorem}{Theorem}[section]
 
 \newtheorem{lemma}[theorem]{Lemma}

\makeatletter
\def\GrabProofArgument[#1]{ #1: \egroup\ignorespaces}
\def\proof{\noindent\textbf\bgroup Proof%
	\@ifnextchar[{\GrabProofArgument}{. \egroup\ignorespaces}}

\makeatother

\usetikzlibrary{arrows,shapes,snakes,automata,backgrounds,petri,calc}
\usepackage[latin1]{inputenc}

\newcommand{\shift}{\mathbb{S}}

\newcommand{\wildcard}{\mathsf{?}}

\newcounter{proccnt}

\newcommand{\konote}[1]{}

\usepackage[normalem]{ulem} 

\title{Quantum Pattern Matching with Wildcards}

\author{
    Masoud Seddighin
	\and Saeed Seddighin
}

\SetCommentSty{mycommfont}

\SetKwInput{KwInput}{Input}                
\SetKwInput{KwOutput}{Output}  
\begin{document}
	\newcommand{\ignore}[1]{}
\renewcommand{\theenumi}{(\roman{enumi}).}
\renewcommand{\labelenumi}{\theenumi}
\sloppy
\date{}

\maketitle

\thispagestyle{empty}

\begin{abstract}
Pattern matching is one of the fundamental problems in Computer Science. Both the classic version of the problem as well as the more sophisticated version where wildcards can also appear in the input can be solved in almost linear time $\tilde O(n)$ using the KMP algorithm and Fast Fourier Transform, respectively. In 2000, Ramesh and Vinay~\cite{ramesh2003string} give a quantum algorithm that solves classic pattern matching in sublinear time and asked whether the wildcard problem can also be solved in sublinear time? In this work, we give a quantum algorithm for pattern matching with wildcards that runs in time $\tilde O(\sqrt{n}\sqrt{k})$ when the number of wildcards is bounded by $k$ for $k \geq \sqrt{n}$. This leads to an algorithm that runs in sublinear time as long as the number of wildcards is sublinear.
\end{abstract}
\section{Introduction}
String problems are very well-studied in computer science both in classic and quantum settings~\cite{ramesh2003string,le2023quantum,boroujeni2021approximating,hajiaghayi2019approximating,boroujeni2019improved,rubinstein2023approximation,jin2022quantum,hajiaghayi2021string,clifford2007simple,hajiaghayi2019massively,boroujeni2020improved,boroujeni20191+,seddighin20223+,bathie2024pattern,clifford2010pattern,clifford2011black}. Pattern matching is perhaps the most fundamental string problem with widespread applications in areas such as text processing, bioinformatics, and data analysis. In the quantum computing landscape, it has also emerged as a critical problem for evaluating the potential speedups offered by quantum algorithms over classical approaches. While existing quantum algorithms provide meaningful improvements over classical methods for classic pattern matching, no speedup is known for quantum algorithms when wildcards are allowed to be in the pattern and the text.

Both the classic version of the problem as well as the more sophisticated version where wildcards can also appear in the input can be solved in almost linear time $\tilde O(n)$ using the KMP algorithm~\cite{cormen2022introduction} and Fast Fourier Transform~\cite{clifford2007simple}, respectively. In 2000, Ramesh and Vinay~\cite{ramesh2003string} give a quantum algorithm that solves classic pattern matching in sublinear time and asked whether the wildcard problem can also be solved in sublinear time? In this work, we give a quantum algorithm for pattern matching with wildcards that runs in time $\tilde O(\sqrt{n}\sqrt{k})$ when the number of wildcards is bounded by $k$ for $k \geq \sqrt{n}$. This yields an algorithm that runs in time $\tilde O(\sqrt{n\max\{k,\sqrt{n}\}})$\footnote{It is easy to show that one can increase the number of wildcards without changing the nature of the problem. For instance, one can replace each character \textsf{x} of the two strings by two characters \$\textsf{x} in both strings where \$ is a special character that does not occur in any of the original strings. Now, one can arbitrarily turn each of the \$ characters of the pattern into a wildcard without changing the complexity of the problem as long as at least one \$ character remains intact.}.
\section{Preliminaries}
We denote the two strings by $A$ (representing the text) and $B$ (representing the pattern). We denote the size of $A$ by $n$ and the size of $B$ by $m$ and assume without loss of generality that $n \geq m$. Each string is a sequence of characters and wildcards that are indexed from 0. We denote the total number of wildcards in both $A$ and $B$ by $k$. Also, we refer to a wildcard character by `$\wildcard$' (without quotes).

We say two characters match, if either they are equal or one of them is a wildcard otherwise we call them a mismatch. We say a string matches another string if they have equal sizes and their characters match index by index. Also, we denote the number of mismatches between two strings $X$ and $Y$ of equal size as the number of $i$'s such that $X_i$ does not match $Y_i$. 

Finally, for a string $X$, we denote by $X[\alpha, \beta]$ a substring of $X$ that starts from index $\alpha$ and ends at index $\beta$.
\section{$\tilde O(\sqrt{n}\sqrt{k})$ Time Algorithm for $k \geq \sqrt{n}$}
In this section, we present a quantum algorithm for pattern matching with wildcards that runs in time $\tilde O(\sqrt{n}\sqrt{k})$ when the number of wildcards is at least $\sqrt{n}$. We make two assumptions here that we will address at the end of this section: (i) We assume that $n/2 < m \leq n$ and (ii) we assume that the value of $k$ is known to us in advance. Also, since quantum algorithms are inherently vulnerable to errors, whenever we say an algorithm can solve a problem in a given time, we mean that algorithm can solve the problem in that time with probability at least $1-n^{-c}$ for any constant $c$. Since we use the $\tilde O$ notation that suppresses the log factors, the constant $c$ of the success probability can be made arbitrarily large and thus the failure scenarios can be ignored.   

We start by defining a measurement of the strings that we call shifted matching. For a shift $1 \leq d < m$, define the \textit{shifted matching array} of $B$, denoted by $\shift(B,d)$ as a 0/1 array of size $m-d$ in the following way:
\begin{itemize}
	\item $\shift(B,d)_i = 0$ if $B_i = B_{i+d}$ and $B_i \neq \wildcard$.
	\item $\shift(B,d)_i = 1$ if either $B_i = \wildcard$ or $B_{i+d} = \wildcard$ or $B_i \neq B_{i+d}$.
\end{itemize}
We also refer to $\sum_{i=0}^{m-d-1} \shift(B,d)_i$ as the \textit{shifted matching sum} of $B$ with shift $d$. The shifted matching arrays and the shifted matching sums of $A$ are defined analogously. We consider two cases in our algorithm.
\begin{enumerate}[label=\arabic*.]
	\item For each $1 \leq d < k$, the shifted matching sum of $B$ with shift $d$ is at least $3k$.
	\item There exists a $1 \leq d < k$ such that the shifted matching sum of $B$ with shift $d$ is bounded by $6k$. 
\end{enumerate}
While the two cases are not necessarily disjoint, for any instance of the problem, at least one of the cases holds. We first state an observation from ~\cite{brassard2000quantum} that enables us to find which case holds for our problem instance in time $\tilde O(\sqrt{n})$ and then solve the problem for each case separately.

\begin{observation}[from ~\cite{brassard2000quantum}]\label{obs:sampling}
	Given a 0/1 array of length $\alpha$, one can determine in time $\tilde O(\sqrt{\alpha / \beta})$ whether the sum of the elements of the array is in range $[0, \beta]$ or in range $[2\beta, \alpha]$. If the answer is in range $(\beta, 2\beta)$ the output of the algorithm could be either case.
\end{observation}

It follows from Observation~\ref{obs:sampling} that we can determine in time $\tilde O(\sqrt{n/k})$ whether for a given $d$, the shifted matching sum of $B$ with shift $d$ is bounded by $3k$. We can then use Grover's algorithm~\cite{grover1996fast} to find one $d$ in range $[1, k)$ such that this shifted sum is bounded by $3k$. If such a $d$ exists, then we can be sure that (incorporating the multiplicative error of the sampling algorithm) the shifted matching sum of $B$ with shift $d$ is certainly bounded by $6k$ (case 2). If no such $d$ exists, then we can be sure that for all $1 \leq d < k$, the shifted matching sum of $B$ with shift $d$ is more than $3k$ (case 1). Since the Grover's algorithm imposes a multiplicative overhead of $\tilde O(\sqrt{k})$ then the overall runtime of the algorithm would be $\tilde O(\sqrt{n})$.

Based on what discussed above, in the first step of our algorithm we spend time $\tilde O(\sqrt{n})$ to find out which case holds for our problem instance and in the second step of the algorithm we solve the problem for that case.

\subsection{Case 1}
In this case, we are sure that for any $1 \leq d < k$, the shifted matching sum of $B$ with shift $d$ is at least $3k$. The key observation that we prove in this section is the following: For each $0 \leq \alpha < \beta \leq n-m$ such that $\beta - \alpha < k$, either the number of mismatches between $A[\alpha, \alpha+m-1]$ and $B$ is at least $k/2$ or the number of mismatches between $A[\beta, \beta+m-1]$ and $B$ is at least $k/2$.

\begin{lemma}\label{lemma:1}
	For each $0 \leq \alpha < \beta \leq n-m$ such that $\beta - \alpha < k$, either the number of mismatches between $A[\alpha, \alpha+m-1]$ and $B$ is at least $k/2$ or the number of mismatches between $A[\beta, \beta+m-1]$ and $B$ is at least $k/2$.
\end{lemma}
\begin{proof}
Assume for the sake of contradiction that the number of mismatches is smaller than $k/2$ in both cases. Define $d = \beta - \alpha$ and consider the shifted matching array of $B$ with shift $d$. For each $0 \leq i < m-d$ such that $\shift(B,d)_i = 1$ one of the following four cases can happen:
\begin{enumerate}[label=\arabic*.]
	\item  $B_i = \wildcard$ or $B_{i+d} = \wildcard$
	\item $A_{\beta+i} = A_{\alpha+i+d} = \wildcard$
	\item There is a mismatch between $A_{\alpha+d+i} = A_{\beta+i}$ and $B_i$. This will be counted as a mismatch between $A[\beta, \beta+m-1]$ and $B$. 
	\item There is a mismatch between $A_{\alpha+d+i} = A_{\beta+i}$ and $B_{i+d}$. This will be counted as a mismatch between $A[\alpha, \alpha+m-1]$ and $B$. 
\end{enumerate}
Since the total number of wildcards in both strings is bounded by $k$, then the total number of $i$'s for which either case 1 or case 2 happens is bounded by $2k$. Moreover, since the number of mismatches between $A[\alpha, \alpha+m-1]$  and $B$ is smaller than $k/2$ and the number of mismatches between $A[\beta, \beta+m-1]$ and $B$ is smaller than $k/2$, the total number of $i$'s for which either case 3 or case 4 happens is smaller than $k$. This implies that the shifted matching sum of $B$ with shift $d$ is smaller than $3k$ which contradicts our original assumption.
\end{proof}

Lemma~\ref{lemma:1} implies that for each interval $[\alpha, \beta]$ of $A$ such that $\beta - \alpha < k$, there is at most one $\alpha \leq i \leq \beta$ such that the number of mismatches between $A[i, i+m-1]$ and $B$ is smaller than $k/2$. In order to solve the problem, we divide the interval $[0, n-1]$ into $\lceil n/k \rceil$ intervals of size at most $k$. We then design an algorithm $f$ that takes an interval $[\alpha, \beta]$ as input such that $\beta - \alpha < k$ and finds out if there is a match between $A[i, i+m-1]$ and $B$ for some $\alpha \leq i \leq \beta$ in the following way:

For a given $\alpha \leq i \leq \beta$, we can run the sampling algorithm of Observation~\ref{obs:sampling} to find out if the number of mismatches between $A[i, i+m-1]$ and $B$ is smaller than $k/4$, or at least $k/2$ in time $\tilde O(\sqrt{n/k})$. We know that the latter is the case for at most one $i$ in range $[\alpha, \beta]$ due to Lemma~\ref{lemma:1}. Therefore, we can run the Grover's algorithm to find such an $i$ (if exists) in time $\tilde O(\sqrt{n})$. If no such $i$ exists, then there is no match between $A[i, i+m-1]$ and $B$ for any $\alpha \leq i \leq \beta$, otherwise let $i$ be the index of the only starting point in range $[\alpha, \beta]$ whose number of mismatches with $B$ is smaller than $k/2$. We can run a Grover's algorithm in time $O(\sqrt{n})$ to find out if $A[i, i+m-1]$ matches with $B$. Thus, the overall runtime for the entire interval is $\tilde O(\sqrt{n})$. Finally, we run another Grover's search over all intervals to find out if there is a solution in any of the intervals. Since the number of intervals is $O(n/k)$, the overhead of the Grover's algorithm is $\tilde O(\sqrt{n/k})$ and therefore the runtime of the algorithm is $\tilde O(n/\sqrt{k})$  in total which is bounded by $\tilde O(\sqrt{n}\sqrt{k})$ for $k \geq \sqrt{n}$.

\subsection{Case 2}
In this case, we know that for a given $1 \leq d < k$, the shifted matching sum of $B$ with shift $d$ is bounded by $6k$. Also, such a $d$ is provided to us.  In what follows, we show two key properties of the shifted matching array of $A$ that play important roles in our algorithm.

\begin{lemma}\label{lemma:2}
	Let for some $0 \leq i \leq n-m$, $A[i, i+m-1]$ match with $B$. Then we have $\sum_{j=i}^{i+m-d-1}\shift(A, d)_j \leq 8k$.
\end{lemma}
\begin{proof}
	Assume for the sake of contradiction that $A[i, i+m-1]$ matches with $B$ for some $0 \leq i \leq n-m$ and also $\sum_{j=i}^{i+m-d-1}\shift(A, d)_j > 8k$ holds. Out of the 1s in the shifted matching array of $A$ with shift $d$ at most $2k$ of them correspond to wildcards and the rest are due to having different characters. Thus, there are at least $6k+1$ distinct values for $j$ in range $[0, m-d-1]$ such that $A_{i+j} \neq A_{i+j+d}$ and neither character is a wildcard. In order for $A[i, i+m-1]$ to match with $B$ we have to have one of the following for each such $j$:
	\begin{itemize}
		\item $B_j \neq B_{j+d}$
		\item $B_j = \wildcard$
		\item $B_{j+d} = \wildcard$
	\end{itemize}
	Notice each of the above cases implies $\shift(B, d)_j = 1$ which means the shifted matching sum of $B$ with shift $d$ is more than $6k$. This contradicts our original assumption.
\end{proof}

\begin{lemma}\label{lemma:3}
	For an $0 \leq i \leq n-m$, $A[i, i+m-1]$ matches with $B$ if and only if all of the following hold:
	\begin{itemize}
		\item For each $0 \leq j < d$,  $A_{i+j}$ matches with $B_j$.
		\item For each $0 \leq j < m-d$ such that $\shift(A,d)_{i+j} =1$,  $A_{i+j}$ matches with $B_j$ and  $A_{i+j+d}$ matches with $B_{j+d}$.
		\item For each $0 \leq j < m-d$ such that $\shift(B,d)_j =1$,  $A_{i+j}$ matches with $B_j$ and $A_{i+j+d}$ matches with $B_{j+d}$.
	\end{itemize}
\end{lemma}
\begin{proof}
	It follows from definition that if one of the conditions does not hold, then $A[i, i+m-1]$ does not match with $B$. We show here that if all of the above conditions hold, then $A[i, i+m-1]$ must match with $B$. Assume for the sake of contradiction all of the above hold but there is a mismatch between $A[i, i+m-1]$ and $B$. Let $j$ be the smallest index for which $A_{i+j}$ and $B_j$ do not match. Neither character can be a wildcard here and also $j \geq d$ should hold otherwise the first condition of the lemma fails. Since $j$ is the smallest such element, we know that $A_{i+j-d}$ matches with $B_{j-d}$. This implies that one of $\shift(A,d)_{i+j-d}$ or $\shift(B,d)_{j-d}$ should be equal to $1$ which means that one of the conditions of the lemma should fail. This is contradiction and thus the lemma holds.
\end{proof}

Lemmas~\ref{lemma:2} and~\ref{lemma:3} give us a convenient tool to find out if $A[i,i+m-1]$ matches with $B$ for any $0 \leq i \leq n-m$. Recall our assumption in the beginning of the section that $n/2 < m \leq n$ holds. This means that  $A_{m-1}$ is present in any interval of size $m$ of $A$. We define $\beta$ as the last element of $A$ that can potentially contribute to a match between $A$ and $B$ according to Lemma~\ref{lemma:2}. To this end, we list up to $8k+1$ indices $i \geq m-1$ of $A$ such that $\shift(A,d)_i = 1$. If fewer than $8k+1$ such indices exists we set $\beta = n-1$, otherwise we list the first $8k+1$ such elements and set $\beta$ equal to the index of the $8k+1$'th such element minus one plus $d$.  Similarly, we define $\alpha$ as the smallest element of $A$ that can contribute to a match from $A$ to $B$ according to Lemma~\ref{lemma:3}. To this end, we list up to $8k+1$ indices $i \leq m-1$ such that $\shift(A,d)_i = 1$. If fewer than $8k+1$ such elements exist then we set $\alpha = 0$, otherwise we list the largest $8k+1$ of those elements and  set $\alpha$ equal  to the index of the smallest such index plus one. It follows from Lemma~\ref{lemma:2} that any match from $A$ to $B$ should only include elements of $A[\alpha, \beta]$ and moreover, we have already listed all elements $i$ in range $[\alpha, \beta-d]$ such that $\shift(A,d)_i = 1$. This process takes time $\tilde O(\sqrt{n} \sqrt{k})$ via the element listing algorithm~\cite{grover1996fast}. We similarly, list all elements of the shifted matching array of $B$ with shift $d$ that are equal to 1. Since their count is bounded by $6k$, the overall runtime would be bounded by $\tilde O(\sqrt{n} \sqrt{k})$. At this point, for each index $\alpha \leq i \leq \beta-m$ we can use Lemma~\ref{lemma:3} to find out if $A[i,i+m-1]$ matches with $B$ in time $\tilde O(\sqrt{k})$ using the Grover's algorithm. Thus, if we run another Grover's algorithm to find out if such a match exists for any  $\alpha \leq i < \beta-m+1$, the overall runtime would be bounded by $\tilde O(\sqrt{n} \sqrt{k})$.

\subsection{Algorithm}
We discussed previously that we consider two cases and solve each case separately in time $\tilde O(\sqrt{n}\sqrt{k})$. Here, we address the two assumptions we made earlier. The first assumption is regarding the value of $k$. Although $k$ is not known to us in advance, we can approximate it in time $\tilde O(\sqrt{n})$ within a multiplicative factor of $2$. In other words, we can determine in time $\tilde O(\sqrt{n})$  via Observation~\ref{obs:sampling} a value $k'$ such that $k \leq k' \leq 2k$ holds with high probability. Notice that all the above arguments continue to hold if we use an upper bound $k'$ instead of the exact value $k$ in our algorithm. The only impact of this change is that the runtime would grow to $\tilde O(\sqrt{n}\sqrt{k'})$ which is asymptotically equal to  $\tilde O(\sqrt{n}\sqrt{k})$.

To address the assumption regarding the sizes of $n$ and $m$, we do the following: For cases that $m \leq n/2$, we construct $\lceil n/m \rceil$ instances of the problem with strings $A^i$ and $B^i$ for each instance $1 \leq i \leq \lceil n/m \rceil$ in the following way:
\begin{itemize}
	\item $B^i = B$ for all instances.
	\item $A^i = A[(i-1)m, (i-1)m+2m-2]$ for $i < \lceil n/m \rceil$.
	\item $A^i = A[n-(2m-1), n-1]$ for $i = \lceil n/m \rceil$.
\end{itemize}
Since each interval of size $m$ of $A$ appears in at least one of the $A^i$'s, then if there is a match between $A$ and $B$, there is certainly a match in one of the problem instance as well. Moreover, the size of the strings in each problem instance is $O(m)$ and they also meet the assumption $|A^i|/2 < |B^i| \leq |A^i|$ and therefore each instance can be solved in time $\tilde O(\sqrt{m}\sqrt{k})$. Therefore, if we run a Grover's search over all instances, the overall runtime would be equal to $\tilde O(\sqrt{n/m}\sqrt{m}\sqrt{k}) = \tilde O(\sqrt{n}\sqrt{k})$.

\begin{theorem}
	There is a quantum algorithm for pattern matching with wildcards that runs in time $\tilde O(\sqrt{n}\sqrt{k})$ and succeeds with high probability when the number of wildcards is at least $\sqrt{n}$.
\end{theorem} 

\bibliographystyle{plainnat}
\bibliography{draft}
\appendix

\end{document}